\documentclass[a4paper,11pt]{article}

\usepackage[utf8]{inputenc}
\usepackage[english]{babel}
\usepackage{fullpage}
\usepackage{amsmath,amsthm,amssymb}
\usepackage{xspace}

\newtheorem{theorem}{Theorem}

\newtheorem{lemma}{Lemma}

\newtheorem{open.problem}{Open Problem}

\newcommand{\Res}{\ensuremath{\mathsf{Res}}\xspace}
\newcommand{\tRes}[1]{\ensuremath{\mathsf{Res}^{*}(#1)}\xspace}

\newcommand{\tReslog}{\ensuremath{\mathsf{Res}^{*}(\mathsf{polylog})}\xspace}

\begin{document} 
\centerline{\bf \Large Short $\tReslog$ refutations if and only if narrow \Res refutations}
\centerline{(an answer to a question of Neil Thapen\footnote{Apparently the present result can be proved by methods of Bounded Arithmetic. He asked whether there exists a simpler and more direct proof.})}
\bigskip
\noindent Massimo Lauria --- \texttt{lauria.massimo@gmail.com}
\newline
\noindent January 28, 2012
\newline
\noindent Updated: \today
\bigskip

In this note we show that any $k$-CNF which can be refuted by a quasi-polynomial \tReslog refutation has a ``narrow'' refutation in \Res (i.e., of poly-logarithmic width).  Notice that while $\tReslog$ is a complete proof system, this is not the case for \Res if we ask for a narrow refutation.  In particular is not even possible to express all CNFs with narrow clauses.  But even for constant width CNF the former system is complete and the latter is not (see for example~\cite{bg01}).  We are going to show that the formulas ``left out'' are the ones which require large $\tReslog$ refutations.
We also show the converse implication: a narrow Resolution refutation can be simulated by a short $\tReslog$ refutation.

\section*{References}

The author does not claim priority on this result. The technical part of this note (Lemmata~\ref{lmm:h} and~\ref{lmm:e}) bears similarity with the relation between $d$-depth Frege refutations and tree-like $d+1$-depth Frege refutations outlined in~\cite{kra1994lower}. Part of it had already been specialized to \Res\ and $\Res(k)$ in~\cite{egm04}.

\section*{Preliminaries}

We consider a formula $F$ in CNF form with $m$ clauses and $n$ variables. We fix $k$ to be the width (i.e., number of literals) of the largest clause in $F$.  The $i$-th clause of $F$ is denoted as $C_{i}$, so that we can write $F$ as $\bigwedge_{i=1}^{m} C_{i}$.

A \emph{refutation} of size $S$ for formula $F$ is a sequence of formulas $D_{1}, \ldots, D_{S}$ such that $D_{i}=C_{i}$ for $1 \leq i \leq m$, and any $D_{i}$ for $i>m$ is either an \emph{axiom} or is \emph{logically inferred} from two previous formulas in the sequence.

Different proof systems are characterized by the types formulas allowed as proof lines, by the axioms and by the logical inference rules. Further constraints on the structure of the refutation may be imposed: while this may cause refutations to be longer, it facilitates the process of searching for the refutation itself.  A customary constraint is to impose \emph{tree-like} structure on the refutation: with the exception of the axioms and the initial clauses of $F$, any formula in the sequence can be used at most once as premise for a logical inference.  In this way the refutation looks like a binary tree in which every leaf is labelled either by an axiom or by an initial clause, and every internal vertex correspond to a formula inferred in the refutation.  Notice that in a tree-like refutation any formula that is needed more than once must be re-derived from scratch.  A proof system is called \emph{dag-like} is no constraint on the structure of the proof is imposed (indeed the refutation looks like a directed acyclic graph).

\textbf{\Res}: every line in the refutation is a clause (i.e., disjunction of literals), and there is a single inference rule:
$$
\frac{A \vee x \qquad B \vee \neg{x}}{A \vee B}
$$

The \emph{width} of a \Res refutation is the maximum number of literals contained in a clause of the refutation. In this note we focus on refutations with poly-logarithmic width.

\textbf{$\tRes{l}$}: the structure of the proof must be tree-like, every line in the refutation is a $l$-DNF, there is an axion introduction rule:
$$
\frac{}{l_{1}l_{2}\cdots l_{s} \bigvee \neg{l_{1}} \vee \neg{l_{2}} \vee \ldots \vee \neg{l_{s}}} \quad \text{for $1 \leq s \leq l$,}
$$
and an inference rule
$$
\frac{A \vee l_{1}l_{2}\cdots l_{s} \qquad B \vee \neg{l_{1}} \vee \neg{l_{2}} \vee \ldots \vee \neg{l_{s}}}{A \vee B} \quad \text{for $1 \leq s \leq l$,}
$$
for any set of $s$ literals $l_{1}, l_{2}, \ldots, l_{s}$.

\medskip

\textbf{Notes on the definition}: there are different ways to define $\tRes{l}$, which are all equivalent up to a polynomial size increase in the size of the refutation.  Notice that there are neither weakening rule nor AND introduction rules: this allows to control how terms  appear in the refutation and this makes the next proofs simpler.  This is without loss of generality \textbf{since we are only dealing with refutations of CNFs}.

The cut rule is defined to require exactly the same number of literals on both sides.  This is more rigid than usual, but \textbf{since the system is tree-like} this also is without loss of generality.  This last restriction is not needed in the following proofs, but is kept to simplify notations.

\section*{Main statement}

In this note we prove that

\begin{theorem}\label{thm:1}
  Let $F$ be a $k$-CNF. $F$ has a quasi-polynomial size $\tReslog$ refutation if and only if has a $\Res$ refutation of poly-logarithmic width.
\end{theorem}

Notice that any dag-like $\Res$ refutation of poly-logarithmic width has size at most quasi-polynomial.

\section*{Proof of narrow $\Res$ simulation}

The proof is based on the following Lemma, which immediately implies one direction of Theorem~\ref{thm:1}.

\begin{lemma}\label{lmm:h}
Let $F$ be a $k$-CNF\@. If $F$ has a refutation in $\tRes{l}$ with $L$ leaves, then has a $\Res$ refutation of width $l \lceil \log L \rceil + max\{k,l\}$.
\end{lemma}

\begin{proof}
  The proof of the lemma is by induction on the number of leaves $L$ in the refutation.  If $L=1$ the initial CNF contains the empty clause, and the result is trivial.

\medskip

Let us assume $L>1$, and fix $w=l \lceil \log L\rceil + max\{k,l\}$. Consider the last step of the refutation.  Since it results in an empty DNF, it must be the result of a cut rule between formulas $C=\bigwedge_{i=1}^{s}l_{i}$ (conjunction) and $D=\bigvee_{i=1}^{s}\neg{l_{i}}$ (disjunction), for a set of $s$ literals.

Since the refutation is tree-like the two proofs of $C$ and $D$ are disjoint, and the number of leaves in each proofs ($L_{C}$ and $L_{D}$ respectively) are such that $L_{C}+L_{D}=L$.  Thus either $L_{C}$ or $L_{D}$ is less than or equal to $\frac{L}{2}$, and both are less than $L$.  The proof is divided in two cases:

\bigskip

\textbf{($L_{C} \leq L/2$)}: Since $F \vdash \bigwedge_{i=1}^{s}l_{i}$ in $\tRes{l}$ with $L_{C}$ leaves, by fixing $l_{i}=0$ we get that $F|_{l_{i}=0} \vdash \square$ is proved by a $\tRes{l}$ refutation which uses at most $L_{C}$ leaves. By inductive hypothesis the same refutation can be done in $\Res$ with width at most $l \lceil \log L_{C}\rceil + max\{k,l\} \leq l \lceil\log L\rceil + max\{k,l\} - l \leq w-1$.

By weakening we have that $F \vdash l_{i}$ in $\Res$ in width $w$, for any $1 \leq i \leq s$. Using such literals we can deduce $F|_{l_{1}=1,\ldots,l_{s}=1}$ from $F$ in width $k$ by removing any occurrences of literals $\neg{l_{i}}$.  Since $F \vdash \bigvee_{i=1}^{s}\neg{l_{i}}$ in $\tRes{l}$, we can prove $F|_{l_{1}=1,\ldots,l_{s}=1} \vdash \square$ in $\tRes{l}$ with at most $L_{C} < L$ leaves.  By inductive hypothesis this refutation can be done in $\Res$ in width $w$.  Composing the resolution proofs $F \vdash l_{i}$ for $1 \leq i \leq s$, the proof of $F,l_{1},\ldots l_{s} \vdash F|_{l_{1}=1,\ldots,l_{s}=1}$ and the proof $F|_{l_{1}=1,\ldots,l_{s}=1} \vdash \square$, we get a $\Res$ refutation of width $w$ of $F \vdash \square$.

\bigskip

\textbf{($L_{D} \leq L/2$)}: we may assume that $s>1$ because otherwise formulas $C$ and $D$ can be swapped and the reasoning for the previous case applies. Since $F \vdash \bigvee_{i=1}^{s}\neg{l_{i}}$ in $\tReslog$ with $L_{D}$ leaves we get that $F|_{l_{1}=1,\ldots,l_{s}=1} \vdash \square$, with a $\tRes{l}$ refutation with $L_{D} \leq \frac{L}{2}$ leaves. By inductive hypothesis the same refutation can be done in $\Res$ with width at most $l \lceil \log L_{D} \rceil + max\{k,l\} \leq l \lceil \log L \rceil + max\{k,l\} - l \leq w-l$.  By weakening $\Res$ proves $F \vdash \bigvee_{i=1}^{s}\neg{l_{i}}$ in width $w$.

We now conclude  arguing that $\Res$ proves $F, \bigvee_{i=1}^{s}\neg{l_{i}} \vdash \square$ with width $w$.  To see that observe the $\tRes{l}$ proof of $F \vdash \bigwedge_{i=1}^{s}l_{i}$: each occurrence of $\bigwedge_{i=1}^{s}l_{i}$ is introduced in such proof using the axiom $\bigwedge_{i=1}^{s}l_{i} \vee \bigvee_{i=1}^{s}\neg{l_{i}}$. Substitute such axiom with the new initial clause $\bigvee_{i=1}^{s}\neg{l_{i}}$.  By an easy induction along the tree-like derivation, such transformation produces a $\tRes{l}$ proof of $F,\bigvee_{i=1}^{s}\neg{l_{i}} \vdash \square$ with $L_{C} < L$ leaves.  By inductive hypothesis this implies a $\Res$ refutation of width $w$ (notice that the initial width of the formula increases, but that is accounted in the definition of $w$).

\end{proof}

\section*{Proof of the short $\tRes{l}$ simulation}

The following lemma gives the other direction of Theorem~\ref{thm:1}.

\begin{lemma}\label{lmm:e}
Let $F$ be any CNF\@. If $F$ has a $\Res$ refutation of width $w$ and size $S$, then $F$ has a $\tRes{w}$ refutation of size $O(S)$.
\end{lemma}

\begin{proof}
  Consider the $\Res$ refutation $D_{1},D_{2},\ldots D_{S}$ of $F$. We define the sequence of $w$-DNFs $E'_{t}=\bigvee_{i=1}^{t} \neg{D_{i}}$.  By backward induction on $t$ from $S-1$ to $0$ we are going to derive a $w$-DNF $E_{t}$ such that the terms of $E_{t}$ are a subset of the terms of $E'_{t}$.  Since $E'_{0}$ is the empty DNF that would conclude the proof.

For $t=S-1$ notice that $E'_{S-1}$ contains $x \vee \neg{x}$ for some variable $x$ in $F$, which is an axiom in $\tRes{w}$.

Fix $t < S - 1$ and consider $D_{a}$ and $D_{b}$ which has been used to derive $D_{t+1}$. For convenience write as follows
$$
D_{a} \equiv A \vee x \quad D_{b}\equiv B \vee \neg{x} \quad D_{t+1} \equiv A \vee B \quad E_{t+1}\equiv \Delta \vee (\neg{A} \wedge \neg{B})
$$
for some $w$-DNF $\Delta$, some clauses $A$, $B$ and some variable $x$.  Terms of $\Delta$ are all contained in $E'_{t+1}$ by inductive hypothesis, and furthermore they are contained in $E'_{t}$ because $\neg{D_{t+1}}$ has been factored out.  Employ the following tree-like deduction
\begin{align}
  (\neg{A} \wedge \neg x) \vee A \vee x       & & \text{Axiom of $\tRes{w}$}\label{eq:1}\\
  (\neg{B} \wedge x     ) \vee B \vee \neg{x} & & \text{Axiom of $\tRes{w}$}\label{eq:2}\\
  (\neg{A} \wedge \neg{x}) \vee (\neg{B} \wedge x) \vee A \vee B & & \text{Cut on term $x$ between~(\ref{eq:1}) and~(\ref{eq:2})}\label{eq:3}\\
  \Delta \vee (\neg{A} \wedge \neg{B}) & & \text{$E_{t+1}$ deduced by induction hypothesis}\label{eq:4}\\
  (\neg{A} \wedge \neg{x}) \vee (\neg{B} \wedge x) \vee \Delta & & \text{Cut on term $\neg{A} \wedge \neg{B}$ between~(\ref{eq:3}) and (\ref{eq:4})}\label{eq:5}
\end{align}
Notice that formula~(\ref{eq:5}) is a $w$-DNF, and its terms are contained in $E'_{t}$.

For $t=0$ we get the empty DNF\@. At each step $E_{t}$ is derived in $\tRes{w}$ using a single occurrence of formula $E_{t+1}$. That means that the whole refutation is tree-like and has $O(S)$ proof lines.
\end{proof}

Notice that the $w$-DNFs have at most $S$ terms each, so the size of the refutation has at most $O(S^{2})$ terms.  For narrow \Res refutations which require large clause space this bound is tight for our construction.


\begin{thebibliography}{EGM04}

\bibitem[BG01]{bg01}
Maria~Luisa Bonet and Nicola Galesi.
\newblock Optimality of size-width tradeoffs for resolution.
\newblock {\em Computational Complexity}, 10(4):261--276, 2001.

\bibitem[EGM04]{egm04}
Juan~Luis Esteban, Nicola Galesi, and Jochen Messner.
\newblock On the complexity of resolution with bounded conjunctions.
\newblock {\em Theor. Comput. Sci.}, 321(2-3):347--370, 2004.

\bibitem[Kra94]{kra1994lower}
Jan Kraj\'{i}\v{c}ek.
\newblock {Lower bounds to the size of constant-depth propositional proofs}.
\newblock {\em Journal of Symbolic Logic}, 59(1):73--86, 1994.

\end{thebibliography}
\end{document}